\documentclass{llncs}
\usepackage{setspace}
\usepackage[numbers,sectionbib]{natbib}
\usepackage{sidecap,floatrow,hyperref,setspace}
\usepackage{algorithm}
\usepackage{algorithmic}
\usepackage{amsmath,amssymb}
\usepackage{graphicx}
\usepackage{epstopdf} 
\usepackage{wrapfig}
\usepackage{lipsum}
\usepackage[affil-it]{authblk}
\usepackage{color,soul}
\usepackage{etoolbox}
\usepackage{tcolorbox}
\usepackage{xcolor}
\makeatletter
\def\@seccntformat#1{\@ifundefined{#1@cntformat}%
   {\csname the#1\endcsname\quad}  
   {\csname #1@cntformat\endcsname}
}
\let\oldappendix\appendix 
\renewcommand\appendix{%
    \oldappendix
    \newcommand{\section@cntformat}{\appendixname~\thesection:\;}
}
\makeatother

\newtheorem{lem}  {Lemma} 
\newcommand {\BL} {\begin{lem}} 
\newcommand {\EL} {\end{lem}} 

\newtheorem{thm} {Theorem} 
\newcommand {\BT} {\begin{thm}}
\newcommand {\ET} {\end{thm}}

\newtheorem{defi} {Definition} 
\newcommand {\BD} {\begin{defi}}
\newcommand {\ED} {\end{defi}}

\newtheorem{pro} {Problem} 
\newcommand {\BP} {\begin{pro}}
\newcommand {\EP} {\end{pro}}

\newtheorem{thmmm} {Property} 
\newcommand {\BPR} {\begin{thmmm}}
\newcommand {\EPR} {\end{thmmm}}

\newtheorem{cor} {Corollary} 
\newcommand {\BCR} {\begin{cor}}
\newcommand {\ECR} {\end{cor}}

\newtheorem{coor} {Observation} 
\newcommand {\BO} {\begin{coor}}
\newcommand {\EO} {\end{coor}}

\newcommand{\lcp}{\mathsf{lcp}}

\newcommand{\trie}{\mathcal{T}}

\newcommand{\s}{\mathcal{S}}
\newcommand{\mr}{\mathsf{maxRun}}

\newcommand{\nd}{\mathsf{nodeDepth}}

\newcommand{\sd}{\mathsf{strDepth}}

\newcommand{\Y}{\mathsf{Y}}
\newcommand{\ACS}{\mathsf{ACS}}
\newcommand{\cha}{\mathsf{char}}
\newcommand{\fre}{\mathsf{freq}}

\newcommand{\X}{\mathsf{X}}

\newcommand{\Dist}{\mathsf{Dist}}

\newcommand{\wt}{\mathsf{weight}}

\title{On Computing Average Common Substring Over Run Length Encoded Sequences}

\author{Sahar Hooshmand$^1$, Neda Tavakoli$^2$, Paniz Abedin$^1$, Sharma V. Thankachan$^1$
}

\date{}
\institute{
$^1$Dept. of Computer Science,  University of Central Florida, Orlando, USA.\\
\email{\{sahar,paniz\}@cs.ucf.edu, sharma.thankachan@ucf.edu
}
\\
$^2$School of Computational Science \& Engineering\\
Georgia Institute of Technology, USA. \\
\email{neda.tavakoli@gatech.edu}}

\begin{document}
\maketitle
\thispagestyle{plain}
\begin{abstract}
The Average Common Substring (ACS) is a popular alignment-free distance measure for phylogeny reconstruction. The ACS of a sequence
$\X[1,x]$ w.r.t.~another sequence $\Y[1,y]$ is
$$\ACS(\X,\Y) = \frac{1}{x} \sum_{i=1}^{x}\max_j \lcp(\X[i,x], \Y[j,y])$$
The $\lcp(\cdot,\cdot)$ of two input sequences is the length of their  longest common prefix. 
The ACS can be computed in $O(n)$ space and time, where 
$n=x+y$ is the input size. 
The compressed string matching is the study of string matching problems with the following twist: the input data is in a compressed format and the underling task must be performed  with little or no decompression. 
In this paper, we revisit the ACS problem under this paradigm where the input sequences are given in their run-length encoded format. We present an algorithm to compute $\ACS(\X,\Y)$  in $O(N\log N)$ time using $O(N)$ space, where $N$ is the total length of sequences after run-length encoding. 
\end{abstract}

\section{Introduction and Related Work}
\label{intro}

The Average Common Substring (ACS), proposed by Burstein {\it et al.}~\cite{RECOMB05}, is a simple alignment-free sequence comparison method. 
This measure and its various extensions~\cite{RECOMB15,ApostolicoGLP16,Leim,Manzini15,Thank1,ThankachanCLAA16,ThankachanCLKA15} have proven to be useful in multiple applications~\cite{apostolico2010maximal, bonham2013alignment, chang2011phylogenetic, comin2012alignment, domazet2009efficient, guyon2009comparison}.
Formally, ACS of a sequence
$\X[1,x]$ w.r.t.~another sequence $\Y[1,y]$, denoted by $\ACS(\X,\Y)$, is
$$\ACS(\X,\Y) = \frac{1}{x} \sum_{i=1}^{x} L[i], \textbf{ where } L[i] = \max_j \lcp(\X[i,x], \Y[j,y])$$
The $\lcp(\cdot,\cdot)$ of two input sequences is the length of their  longest common prefix. 
 The (symmetric) distance based on ACS is~\cite{RECOMB05}:
$$\Dist(\X,\Y) = \frac{1}{2} \bigg(\frac{\log |\Y|}{\ACS(\X,\Y)}+\frac{\log |\X|}{\ACS(\Y,\X)}\bigg)-
\frac{1}{2}\bigg(\frac{\log |\X|}{\ACS(\X,\X)}+\frac{\log |\Y|}{\ACS(\Y,\Y)}\bigg)$$

The computation of ACS is straightforward in $O(n)$ space and time using the generalized suffix tree of $\X$ and $\Y$, where $n=x+y$ is the input size~\cite{RECOMB05}. In this paper, we study the problem of computing ACS, where the input sequences are $\X'[1,x']$ and $\Y'[1,y']$, where $\X'[1,x']$ (resp., $\Y'[1,y']$) is the sequence corresponding to the run-length encoding of $\X[1,x]$ (resp., $\Y[1,y]$). Run-length  encoding is a simple algorithm used for data compression in which runs of data (occurring of the same character on consecutive positions) are stored as a single charter followed by the count of its consecutive occurrences.
The challenge here is to design an   algorithm for computing ACS in space and time close to $O(N)$ instead of $O(n)$, where $N=x'+y'$. We answer this question positively by presenting the following theorem. 

\begin{theorem} \label{main-theorem}
Given two input sequences in their run-length encoded format, the distance between them based on the Average Common Substring (ACS) measure can be computed in $O(N)$ space and $O(N\log N)$ time, where $N$ is the total length of sequences after run-length encoding. 
\end{theorem}

\section{Notation and Background} \label{sec:preli}
Let $\Sigma$ be the alphabet set from which the symbols in $\X$ and $\Y$ are drawn from. 
We denote $\X, \X'$ and $\Y, \Y'$ as follows: 
$$\X = \alpha_1^{f_1} \alpha_2^{f_2} \alpha_3^{f_3} ... \alpha_{x'}^{f_{x'}} \textbf{ and } \X' = (\alpha_1,{f_1}) (\alpha_2,{f_2}) (\alpha_3,{f_3}) ... (\alpha_{x'},{f_{x'}})  $$
$$\Y = \beta_1^{g_1} \beta_2^{g_2} \beta_3^{g_3} ... \beta_{y'}^{g_{y'}} \textbf{ and } \Y' = (\beta_1,{g_1}) (\beta_2,{g_2}) (\beta_3,{g_3}) ... (\beta_{y'},{g_{y'}})$$
Specifically, $\X$ is the concatenation of $f_1$ occurrences of $\alpha_1$ followed by 
$f_2$ occurrences of $\alpha_2$, and so on. Similarly, $\Y$ is the concatenation of $g_1$ occurrences of $\beta_1$ followed by $g_2$ occurrences of $\beta_2$, and so on. Here $\alpha_i, \beta_i \in \Sigma$ and $f_i, g_i \in \{1,2,3,...,n\}$. 
The lexicographic order between two characters $(c,k)$ and $(c',k')$ in the encoded sequences is defined as follows: 
$(c,k)$ is lexicographically smaller than $(c',k')$ iff either $c$ is lexicographically smaller than $c'$ or $c=c'$ and $k <k'$. 
Also, define the suffixes 
$$\X'[i,x'] = (\alpha_i,{f_i}) (\alpha_{i+1},f_{i+1}) ... (\alpha_{x'},{f_{x'}}) \textbf{ and } $$ 
$$ \Y'[i,y'] =  (\beta_i,{g_i}) (\beta_{i+1},g_{i+1})  ... (\beta_{y'},{g_{y'}})$$
If a suffix is a prefix of another suffix, we say that the shortest one is lexicographically smaller. 
Notice that for each $\X'[i,x']$ (resp., $\Y'[i,y']$), there exists an {\bf equivalent suffix} 
$\X[F(i),x]$ (resp., $\Y[G(i),y]$) of $\X$ (resp., $\Y$). Specifically, 
$$F(i) = 1+ \sum_{k<i} f_k  \textbf{ and } G(i) = 1+\sum_{k<i} g_k $$

\BO \label{obs1}
The $k$th lexicographically smallest suffix in $\s$ and the $k$th lexicographically smallest suffix in $\s'$ are equivalent for all values of $k \in [1,N]$, where 
$$\s = \{\X[F(i),x] \mid 1 \leq i \leq x' \} \cup \{\Y[G(i),y] \mid 1\leq i \leq y' \} $$ 
$$\textbf{ and } \s' = \{\X'[i,x'] \mid 1 \leq i \leq x' \} \cup \{\Y'[i,y'] \mid 1\leq i \leq y' \}$$ 
\EO 


The main component of our algorithm is a trie $\trie$ over all strings in $\s$.
It consists of $N$ leaves and at most $N-1$ internal nodes.
Each leaf node in $\trie$ corresponds to a unique suffix in $\s$. 
Specifically, the $i$th leftmost leaf $\ell_i$ corresponds to the $i$th lexicographically smallest suffix in $\s$. 
Each internal node $v$ is associated with two values, (i) $\nd(v)$: the number of nodes on the path from root to $v$ and (ii) $\sd(v)$: the length of the longest common prefix over all suffixes corresponding to the leaves under $v$. Additionally, we call a leaf type-X (resp., type-Y) if the suffix corresponding to it is from $\X$ (resp., $\Y$). The space occupancy of $\trie$ is $O(N)$ words. 

\begin{lemma}
The trie $\trie$ can be constructed in $O(N\log N)$ time using $O(N)$ space. 
\end{lemma}

\begin{proof}
We construct a generalized suffix tree of $\X'$ and $\Y'$ and then convert it into $\trie$~\cite{ST1,ST2,Farach97} by exploring Observation~\ref{obs1}. We defer the details to the full version of this paper. \qed 
\end{proof}



%
%

\section{An $O(N)$-Space and ${O}(n\log N)$-Time Algorithm}

The first step is to construct $\trie$ from $\X'$ and $\Y'$.
Then, we associate each leaf node (except two\footnote{Specifically, the leaves corresponding to $\X[F(x'),x]$ and $\Y[G(y'),y]$.}) with two values, $\cha(\cdot)$ and $\fre(\cdot)$ as follows: 
Let $\ell_a$ be the leaf corresponding to $\X[F(i+1),x]$. Then  
 $$\cha(\ell_a) = \alpha_{i} \text{ and } \fre(\ell_a) = f_{i}$$ 
Similarly, let $\ell_b$ be the leaf corresponding to 
 $\Y[G(j+1),y]$. Then 
 $$\cha(\ell_b) = \beta_{j} \text{ and } \fre(\ell_b) = g_{j}$$
For each $\sigma \in \Sigma$, define (and compute) 
$$\mr(\sigma) = \max\{g_k \mid k \in [1,y'] \textbf{ and } \beta_k = \sigma \}$$
The key intuition behind our algorithm is the following simple observation.

\BO \label{obs2}
Let $\ell_a$ be the leaf in $\trie$ corresponding to the suffix $\X[F(i+1),x]$ for an $i \in [1,x'-1]$. 
Also, let $\X[p,x] = \alpha_{i}^h \circ \X[F(i+1),x] $ for some $h \in [1,f_i]$. 
Specifically, 
$p = F(i+1)-h$ and ``$\circ$" denotes concatenation. 
Then $L[p]$ 
\begin{itemize}
\item  is $\mr(\alpha_{i})$ if $h > \mr(\alpha_{i})$ and  

\item  is $h +\sd(v)$ otherwise, where node $v$ is the lowest ancestor of $\ell_a$ such that there exists a type-Y leaf under $v$ with $\cha(\cdot) = \alpha_{i}$ and $\fre(\cdot) \geq h$.
\end{itemize}
\EO
We now present an efficient algorithm for computing $L[\cdot]$'s based on the above observation. 
First we construct a collection $\{ \trie_{\sigma} \mid \sigma \in \Sigma \}$ of new tries from $\trie$. Specifically, the  $\trie_{\sigma}$ is a compact trie over all those suffixes in $\trie$, such that $\cha(\cdot)$ of the leaves corresponding to them is $\sigma$. 
The total number of nodes over all $\trie_{\sigma}$'s is $O(N)$ as each leaf node in $\trie$ belongs to exactly one $\trie_{\sigma}$. Moreover, they can be extracted from $\trie$ in 
$O(N)$ total time. 

Next, we pre-process each $\trie_{\sigma}$ in time linear to its size for answering level ancestor queries in constant time~\cite{BenderF00}. A level ancestor query $(v,l)$ asks to return the ancestor $u$ of $v$ with $\nd(u) =l$. Finally, for each internal node $v$ in each $\trie_{\sigma}$, 
we compute $\fre(v)$, which is the maximum over $\fre(\cdot)$'s of all type-Y leaves under $v$. 
Note that $\fre(v) =0$ if all leaves under $u$ are of type-X. 
This step can also be implemented in linear time via a bottom up traversal of $\trie_{\sigma}$. 

We are now ready to present the final steps of our algorithm. 
For $p = 1,2,3,...,x$, we compute $L[p]$ as follows. 
Let $\X[p,x] = \alpha_{i}^h \circ \X[F(i+1),x]$.
Then, 
\begin{enumerate}
\item $L[p] = \mr(\alpha_{i})$ if $h > \mr(\alpha_{i})$.
\item Otherwise, find the leaf node $w$ in $\trie_{\alpha_{i}}$ corresponding to  
$\X[F(i+1),x]$ and its  lowest ancestor $v$ such that $\fre(v) \geq h$. 
Fix $v$ as the root when $i =x'$.
Then set $L[p] = h + \sd(v)$. 
The node $v$ can be computed via a binary search (using level ancestor queries) over the nodes on the path to root in $O(\log N)$ time. 
 
\end{enumerate}
This completes the description of our algorithm. The correctness is immediate from Observation~\ref{obs2}.
The total time complexity is $O(n\log N)$.

\section{An $O(N)$-Space and ${O}(N\log N)$-Time Algorithm}
Define $A[i]$ for $i = 1, 2, 3, ..., x'$, where
$$A[i] = \sum_{p=F(i)}^{F(i+1)-1}L[p]$$
Therefore, $$\ACS(\X,\Y) = \frac{1}{x} \sum_{i=1}^{x'} A[i]$$
We now present a new algorithm in which we compute each 
$A[i]$ in $O(\log N)$ time. For each internal node $v$ in $\trie_{\sigma}$, define $\wt(v)$ as follows: $\wt(\cdot)$ of the root node is $0$. For any other node $v$ with $v'$ being its parent,
$$\wt(v) = \wt(v')+ \fre(v) \times (\sd(v)-\sd(v'))$$
By performing a top-down tree traversal, we compute $\wt(\cdot)$ over all internal nodes in $\trie_{\sigma}$ in time linear to its size. Therefore, time over all  $\trie_{\sigma}$'s is $O(N)$. 
We now compute $A[i]$'s using the following steps. 

\begin{itemize}
\item 
For any  $i \in [1,x'-1]$, we first find the leaf node $w$ in
$\trie_{\alpha_{i}}$ corresponding to the suffix $\X[F(i+1),x]$. Also, find the lowest ancestor $v$ of $w$, such that there exits a type-Y leaf under $v$ (equivalently $\fre(v) \neq 0$) via  binary search using level ancestors queries. This step takes $O(\log N)$ time. 
Next, we have two cases and we handle them separately as follows.  
For brevity, let $m = \mr(\alpha_{i})$. 

\begin{itemize}
 \item If $f_{i} > m$, then 
 $$A[i] = \wt(v)+(1+2+3+...+m)+ m(f_{i}-m)$$
By simplifying, we have $A[i] = \wt(v)+m( f_{i} -(m -1)/2)$.

 \item If $f_{i} \leq  m$, then find the lowest ancestor $u$ of $w$, such that $\fre(u) \geq f_{i}$ (via  binary search using level ancestors queries). 
 Then, 
 $$A[i] = \wt(v) - \wt(u) + f_{i} \times \sd(u)+ f_{i}(1+f_{i})/2$$
\end{itemize}
\item 
For $i = x'$, $A[i]$ 
\begin{itemize}
\item is $(1+2+...+f_i) = f_i(f_i+1)/2$ if $f_i \leq m$ and 
\item is $(1+2+...+m)+ m(f_i-m) = m(f_i -(m-1)/2)$, otherwise. 
\end{itemize}

\end{itemize}
In summary, the time complexity is $O(N\log N)$ plus $O(\log N)$ for each $i$ in $[1,x']$. Therefore, total time is $O(N\log N)$. 
The correctness follows from Observation~\ref{obs2} and the definition of $\wt(\cdot)$. This completes the proof of Theorem 1.

\section*{Acknowledgments}

This research is supported in part by the U.S. National Science Foundation under the grant CCF-1703489. 

\bibliographystyle{abbrv} 
\bibliography{references}

\end{document}